\documentclass[english]{cccconf}
\usepackage[comma,numbers,square,sort&compress]{natbib}
\usepackage{epstopdf}

\usepackage{mathtools}  
\usepackage{amssymb}  
\usepackage{amsthm}

\theoremstyle{definition}
\newtheorem{theorem}{Theorem}
\newtheorem{lemma}[theorem]{Lemma}

\newtheorem{proposition}[theorem]{Proposition}

\theoremstyle{definition}
\newtheorem{definition}{Definition}
\newtheorem{remark}{Remark}

\newtheorem{assumption}{Assumption}

\newcommand{\re}{{\mathrm{Re}}}

\newcommand{\im}{{\mathrm{Im}}}
\newcommand{\rr}{{\mathbb{R}}}

\begin{document}

\title{Towards Distributed Stability Analytics of Dynamic Power Systems: A Phasor-Circuit Theory Perspective}
\author{Peng Yang\aref{thu},
        Feng Liu\aref{thu},
        Zhaojian Wang\aref{thu},
    Shicong Ma\aref{JF}}


\affiliation[thu]{Department of Electrical Engineering, Tsinghua University, Beijing 100084, China 
        \email{lfeng@tsinghua.edu.cn}}
\affiliation[JF]{China Electric Power Research Institute, Beijing 100192, China
        \email{mashicong@epri.sgcc.com.cn}}

\maketitle

\begin{abstract}
With the rapid development of renewable and distributed energies, the underlying dynamics of power systems are no longer dominated by large synchronous generators, but by numerous dynamic components with heterogeneous characteristics. In such a situation, the traditional stability analysis method may fail due to the challenges of heterogeneity and scalability. In this paper, we handle this issue by fundamental circuit theory. Inspired by the work of Brayton and Moser in the nonlinear RLC circuit, we extend the concept of the voltage potential to phasor circuits and offer new results into the distributed stability analytics in power systems. We show that under certain distributed passivity-like conditions the system-wide stability can be ensured. The simulation of a 3-bus system is also provided to verify our results.
\end{abstract}

\keywords{Power Systems Stability, Distributed Stability Criterion, Phasor-Circuit Theory, Convexity, Passivity.}

\footnotetext{This work is supported by the project "Research on The Change of Stability Characteristics in Power Systems with High Renewable Energy Penetration".}

\section{Introduction}

Stability is the primary concern in power systems. In recent years, however, with the rapid development of renewable energies and distributed energy technologies, the underlying dynamics of power systems are changing and deterioration of system-wide stability has been witnessed \cite{ts,os}.

Traditionally, the stability analysis only involves the dynamics of synchronous machines and is carried out in a centralized manner \cite{Kundur}. These center-based methods may fail since the dynamics of the power system will be no longer dominated by similar synchronous machines but most-likely consists of numerous heterogeneous dynamic components \cite{path}.
The challenges under such a circumstance are mainly two folds: heterogeneity and scalability. This motivates distributed stability analytic methods which adapt to heterogeneous components' models and can be carried out distributedly while guaranteeing the system-wide stability.

Many efforts have been putting into this task recently. Methods based on distributed analysis of the Jacobian matrix \cite{Song2017,xie,Illic} and the transfer function \cite{ds2,Pates} have been proposed. These methods can analyze the small-signal stability of the interconnected system in by distributed conditions. 
Another approach is based on the concept of passivity or dissipativity \cite{Caliskan,Kasis,Spanias}. In this approach, the system-wide stability is induced as long as each component meets certain passivity condition. The biggest challenge in this approach is to find the right passivity condition to minimize its conservativeness and improve its applicability while maintaining the system-wide stability.
In addition, methods based on 
linear matrix inequalities\cite{xie2}, sum-of-square technique and vector Lyapunov functions \cite{Kundu} are also proposed. These computation-based methods decompose the task of system-wide stability assessment into several distributed calculations. 
The aforementioned methods, however, are still insufficient to meet the urgent requirement in heterogeneous scalable power systems, since they either limit to the small-signal stability or suffer
from the computational burden.

In this paper, we turn to fundamental circuit theories and present another perspective to handle the issue of distributed stability with concerns about heterogeneity and scalability. As we only concern the dynamics near the nominal frequency, the AC power system can be regarded as a circuit in the sinusoidal quasi-steady state \cite{Kundur} and is essentially a phasor circuit as defined in this paper. 
We extend the idea of Brayton and Moser \cite{BM,BM2} , which was first proposed in 1964 to analyze the stability of topologically complete nonlinear RLC circuits, to a class of phasor circuit with special applications to power systems. We first define the voltage potential of the phasor circuit following the line in \cite{BM}. Then we explore its properties with mathematical tools in complex analysis and dynamic systems. Based on that, we provide a passivity-like condition for each component to guarantee the system-wide stability in power systems. Due to the space limit, we present several claims without proof in this paper.

The rest of this paper is organized as follows. Section 2 introduces some basic concepts and formulates the problem; the voltage potential of the phasor circuit is defined and analyzed in Section 3; the distributed stability issue is addressed in Section 4; a numerical example is illustrated in Section 5; and Section 6 concludes this paper.

Notations: $j$ is the imaginary unit; $\rr$ and $\mathbb{C}$ are the sets of real and complex numbers, respectively; $\rr_{\geq0}$ is the set of non-negative real numbers; superscript $^*$ is the complex conjugate; $\text{col}(x_1,x_2)$ is a column vector with entries $x_1$ and $x_2$; $\textbf{0}_n\in\rr^n$ denotes a vector with all zeros entries; for $x\in\mathbb{C}$, $\im x$ and $\re x$ stand for the imaginary and the real part of $x$, respectively.

\section{Problem Formulation}
\subsection{Phasor Representation of Power Systems}
Consider a symmetric AC three-phase power system. Electrical quantities in such a system have the following form.
\begin{definition}\cite{Akagi}
	A function of time $x_{abc}:\rr_{\geq0}\to\rr^3$ is called a symmetric AC three-phase signal if it is described by
	\begin{equation}\label{eq:3phase}
	x_{abc}(t)=\begin{bmatrix}
	x_a(t)\\x_b(t)\\x_c(t)
	\end{bmatrix}=A(t)\begin{bmatrix}
	\sin(\theta(t))\\\sin(\theta(t)-\frac{2\pi}{3})\\\sin(\theta(t)+\frac{2\pi}{3})
	\end{bmatrix}
	\end{equation}
	where $A:\rr_{\geq0}\to\rr_{\geq0}$ is called the amplitude and $\theta:\rr_{\geq0}\to\rr$ is called the phase angle.
\end{definition}
Note that both the amplitude and phase angle may change with time. For the simplicity of notations, we will omit the time argument whenever it is clear in the context. Due to the symmetry in \eqref{eq:3phase}, a coordinate transformation, known as the $dq0$-transformation, is introduced to simplify the analysis.
\begin{definition}\cite{Schiffer}
	Let $x:\rr_{\geq0}\to\rr^3$ and $\phi:\rr_{\geq0}\to\rr$. Consider the matrix function $T_{dq0}:\rr\to\rr^{3\times3}$
	\begin{equation*}
	T_{dq0}(\phi)=\sqrt{\frac{2}{3}}\begin{bmatrix}
	\cos(\phi)&\cos(\phi-\frac{2\pi}{3})&\cos(\phi+\frac{2\pi}{3})\\
	\sin(\phi)&\sin(\phi-\frac{2\pi}{3})&\sin(\phi+\frac{2\pi}{3})\\
	\frac{\sqrt{2}}{2}&\frac{\sqrt{2}}{2}&\frac{\sqrt{2}}{2}
	\end{bmatrix}
	\end{equation*}
	Then, the mapping $f_{dq0}:\rr^3\times\rr\to\rr^3$
	\begin{equation}\label{eq:dq0}
	f_{dq0}(x(t),\phi(t))=T_{dq0}(\phi)x(t)
	\end{equation}
	is called the $dq0$-transformation.
\end{definition}

Apply the $dq0$-transformation to the symmetric AC three-phase signal $x_{abc}$ yields
\begin{equation}
x_{dq0}=\begin{bmatrix}
x_d\\x_q\\x_0
\end{bmatrix}=T_{dq0}(\phi)x_{abc}=\sqrt{\frac{3}{2}}A\begin{bmatrix}
\sin(\theta-\phi)\\\cos(\theta-\phi)\\0
\end{bmatrix}
\end{equation}
Since $x_0(t)\equiv0$, a symmetric AC three-phase signal is totally dictated by its $dq$ components as follows.
\begin{equation}\label{eq:dq}
x_{dq}=\begin{bmatrix}
x_d\\x_q
\end{bmatrix}=\sqrt{\frac{3}{2}}A\begin{bmatrix}
\sin(\theta-\phi)\\\cos(\theta-\phi)
\end{bmatrix}
\end{equation} 
Let a complex number $\bar{X}=\sqrt{\frac{3}{2}}Ae^{j(\theta-\phi)}$. It follows that
\begin{equation}\label{eq:phasordq}
\bar{X}=x_q+jx_d
\end{equation}
We call it the phasor representation of $x_{abc}$ and refer $\bar{X}$ as a phasor, which is defined rigorously as follows.
\begin{definition}
	A function of time $\bar{X}:\rr_{\geq0}\to\mathbb{C}$ is called a phasor if it is described by\begin{equation}\label{eq:phasor}
	\bar{X}(t)=X(t)e^{j\phi(t)}
	\end{equation}
	where $X:\rr_{\geq0}\to\rr_{\geq0}$ is called the magnitude and $\phi:\rr_{\geq0}\to\rr$ is called the phase angle. We also denote $\bar{X}(t):=X(t)\angle\phi(t)$.
\end{definition}
\begin{remark}
	Our definition of phasor follows the line in \cite{Zhong2012}, where the magnitude and angle are both functions of time. Note, however, in the study of the sinusoidal steady-state circuit, a phasor usually means a constant complex number.
\end{remark}

It has been shown in \cite{Schiffer} that the dynamic state variables of a symmetric AC three-phase power system can be expressed as phasors via time-scale separation. This justifies all phasor-based dynamic models for power system stability analysis, such as the well-known network-reduction and network-preserving model \cite{ds}.

\subsection{Power Systems as Phasor Circuits}\label{sec:2.2}
Consider a power system represented in the phasor coordinate as introduced previously. We now show that such a system can be regarded as a circuit with phasor electrical quantities which we call a \textit{phasor circuit} in this paper.

The symmetric AC three-phase power system can be abstracted as a directed graph $G=(\cal V,\cal E)$, where $\cal V$ is the set of nodes and $\cal E$ is the set of branches. Each branch stands for a symmetric three-phase component in the power system. Assume graph $G$ has $b$ branches and $n$ nodes. Among all nodes, one specific node corresponds to the \textit{ground} which serves as a magnitude reference. 
Each branch $\mu\in\cal E$ is associated with a symmetric AC three-phase voltage $v_{abc}^\mu$ and current $i_{abc}^\mu$, as well as their phasor representation $\bar{V}_\mu$ and $\bar{I}_\mu$. We assume the voltages and currents take the associated reference direction\footnote{The associated reference direction means that positive current is defined as flowing into the terminal which is defined to have positive voltage. Note that these directions may be different from the direction of the actual current flow and voltage.} endowed by $G$. 

As components are interconnected electrically, their voltages and currents are constrained by the Kirchhoff's current law (KCL) and the Kirchhoff's voltage law (KVL) as follows. 
\begin{equation}\label{eq:KL}
\sum_{node}\pm i^\mu_{abc}=\textbf{0}_3,\quad\sum_{loop}\pm v^\mu_{abc}=\textbf{0}_3
\end{equation}
where $\pm$ means that the signal takes proper sign accordingly. 
\begin{proposition}
	Consider a symmetric AC three-phase power system $G=(\mathcal{V},\mathcal{E})$. If the $dq0$-transformation \eqref{eq:dq0} with a uniform $\phi(t)$ is applied to each symmetric AC three-phase line current $i_{abc}^\mu$ and voltage $v_{abc}^\mu$, then the corresponding phasors $\bar{I}_\mu$ and $\bar{V}_\mu$ satisfy the following KCL and KVL.
	\begin{equation}\label{eq:KLp}
	\sum_{node}\pm \bar{I}_\mu=0,\quad\sum_{loop}\pm \bar{V}_\mu=0
	\end{equation}
\end{proposition}
\begin{proof}
	Clearly with a uniform $\phi(t)$ we have
	$$\sum\pm x^\mu_{dq0}=\sum\pm T_{dq0}(\phi)x^\mu_{abc}=T_{dq0}(\phi)\sum\pm x^\mu_{abc}=\textbf{0}_3$$
	which completes the proof.
\end{proof}
It shows that if the $dq0$-transformation with a uniform $\phi$ is applied to the power system, then the system can be regarded as a phasor circuit. That means the power system is essentially a graph $G=(\mathcal{V},\mathcal{E})$ with branch currents and voltages as phasors and obeys the fundamental KCL and KVL.
In the associated reference direction, the inner product of branch voltage phasor and negative current phasor defines the \textit{complex power generation} in the corresponding branch,
\begin{equation}\label{eq:S}
S_\mu=-\bar{I}_\mu^*\bar{V}_\mu=P_\mu+jQ_\mu
\end{equation}
where $P_\mu$ and $Q_\mu$ are the active and reactive power generated in branch $\mu$, respectively.

A component in the power system is abstracted as a branch in $G$, which determines the relation between $\bar{V}_\mu$ and $\bar{I}_\mu$ locally. We assume that the power system consists of dynamic voltage sources/loads, transmission lines, and constant power loads. All these components can be divided into two kinds: the static and the dynamic. We denote the set of branches associated with static and dynamic components by $\cal S$ and $\cal D$, respectively. We have $\cal E=\cal S\cup\cal D$.

The static component establishes a mapping between $\bar{V}_\mu$ and $\bar{I}_\mu$, which is also called the voltage-current characteristic in the context of circuit theories. The static component can be generically modeled as
\begin{equation}
g_\mu(\bar{V}_\mu,\bar{I}_\mu)=0
\end{equation}
In this paper, we consider two kinds of static component in a power system. One is the linear admittance $y_\mu$ which is used to model the transmission line in a power system \cite{Kundur}. It gives a static relation as 
\begin{equation}\label{eq:static}
\bar{I}_\mu-y_\mu \bar{V}_\mu=0
\end{equation}
The other is the constant power source/load
\begin{equation}
\bar{I}_\mu^*\bar{V}_\mu-P_\mu^0-jQ_\mu^0=0
\end{equation}
Denote the sets of these two type branches by $\mathcal{S}_1$ and $\mathcal{S}_2$.

The dynamic component relates $\bar{V}_\mu$ and $\bar{I}_\mu$ by differential equations.
We consider a generic model for the dynamic component as follows.
\begin{equation}\label{eq:dynamic}
\dot{x}_\mu=f_\mu(x_\mu,u_\mu)
\end{equation}
where $x_\mu=\text{col}(\xi_\mu,V_\mu,\theta_\mu)\in\mathcal{X}_\mu\times\rr_{>0}\times\rr$ is the state variable of dynamic component $\mu$, and $\xi_\mu\in\mathcal{X}_\mu$ is the auxiliary state variable which includes the heterogeneous dynamics of each component. The input is the power generation in the branch $u_\mu=(P_\mu,Q_\mu)\in\rr^2$. $f_\mu:\;\mathcal{X}_\mu\times\rr_{>0}\times\rr^3\to\mathcal{X}_\mu\times\rr_{>0}\times\rr$ is a continuously differentiable function.
Noting that $\bar{I}_\mu=-(\frac{P_\mu+jQ_\mu}{V_\mu})^*$, component \eqref{eq:dynamic} determines a dynamic relation between $\bar{V}_\mu$ and $\bar{I}_\mu$.

In a power system, the generic formulation \eqref{eq:dynamic} can represent a wide variety of dynamics, such as the synchronous machine \cite{Kundur}, the inverter-interfaced power source in grid-feeding mode \cite{Schiffer}, and loads with frequency and voltage response \cite{Kasis}.
To state our results, we make the following assumption on the circuit topology of a power system.
\begin{assumption}\label{as:1}
	The graph of the circuit is connected. And all dynamic branches have one end connecting to the ground, i.e., they are connected to the power system in a parallel fashion.
\end{assumption}

Assumption \ref{as:1} is usually true for a power system since generators and loads are usually connected to the ground.

Combining all components' relations and the circuit interconnection constrain, the entire phasor circuit can be modeled by a set of differential algebraic equations (DAEs) as follows.
\begin{equation}\label{eq:entire}
\left\lbrace \begin{aligned}
\dot{x}_\mu&=f_\mu(x_\mu,u_\mu),\;\mu\in\mathcal{D}\\
0&=g_\mu(\bar{V}_\mu,\bar{I}_\mu),\;\mu\in\mathcal{S}\\
0&=\sum_{node}\pm \bar{I}_\mu,\;
0=\sum_{loop}\pm \bar{V}_\mu
\end{aligned}\right. 
\end{equation}
\begin{assumption}\label{as:2}
	The DAEs \eqref{eq:entire} of a phasor circuit  is of index-1 \cite{DAE}.
\end{assumption}
Assumption \ref{as:2} is very common in the study of network-preserving power system model \cite{ds}. It ensures the mapping from the state variables to algebraic is one-to-one locally.

\section{Voltage Potential of Phasor Circuits}
\label{sec:2}
Consider a phasor circuit $G$. The set of branch voltages $\bar{V}=(\bar{V}_1,\bar{V}_2,\ldots,\bar{V}_b)$ and the set of branch currents $\bar{I}=(\bar{I}_1,\bar{I}_2,\ldots,\bar{I}_b)$ are vectors in $\mathbb{C}^b$. The KCL and the KVL of phasor circuits \eqref{eq:KLp} put a linear constrain to the domain of $\bar{I}$ and $\bar{V}$. Denote $\mathcal{I}$ and $\mathcal{V}$ as two subsets of $\mathbb{C}^b$ such that every $\bar{I}\in\cal{I}$ and $\bar{V}\in\mathcal{V}$ satisfies \eqref{eq:KLp}. We have both $\cal I$ and $\cal V$ are subspaces of $\mathbb{C}^b$ as \eqref{eq:KLp} are linear. Define the inner product of two vectors in $x,\;y\in\mathbb{C}^b$ as $<x,y>=\sum_{\mu=1}^{b}xy^*$. Following the lines in \cite[Therome 1]{BM} and \cite[Therome 2]{BM}, we have following two lemmas.

\begin{lemma}
	If $\bar{I}\in\cal I$ and $\bar{V}\in\cal V$, then $<\bar{V},\bar{I}>=0$, i.e. $\cal V \perp \cal I$. 
\end{lemma}

\begin{lemma}\label{lem:line}
	Let $\Gamma$ be a one-dimensional curve in $\mathbb{C}^b$, then we have
	$$\int_\Gamma \sum_{\mu\in\mathcal{V}} \bar{I}_\mu^*d\bar{V}_\mu=0$$
\end{lemma}

\subsection{Line Integral of Static Components}
Now consider the phasor circuit composes of static and dynamic components introduced in Section \ref{sec:2.2}.
By Assumption \ref{as:2}, the voltages and currents of all static network branches are determined by the voltages of dynamic branches. So we can specifically choose $\Gamma$ from a fixed point in $\mathbb{C}^b$ to a variable one in such a manner that $\Gamma$ is a solution trajectory of \eqref{eq:entire}.

Along this $\Gamma$, we have
\begin{equation}\label{eq:SLI}
\int_\Gamma  \bar{I}_\mu^*d\bar{V}_\mu=\int_\Gamma  y_\mu^* \bar{V}_\mu^*d\bar{V}_\mu,\;\forall\mu\in\mathcal{S}_1
\end{equation}
As this complex function clearly violates the Cauchy-Riemann equations, the above line integral is path-dependent \cite{complex}. However, Inspecting the real and imaginary part of this integral separately, we have the following proposition.
\begin{proposition}
	The real part of \eqref{eq:SLI} is path-independent if and only if $\text{Im} y_\mu=0$. And the imaginary part of \eqref{eq:SLI} is path-independent if and only if $\text{Re} y_\mu=0$.
\end{proposition}
\begin{proof}
	Denote $y_\mu=g_\mu+jb_\mu$ and $\bar{V}_\mu=x+jy$. Consider the imaginary part of \eqref{eq:SLI} as an example.
	\begin{equation}\label{eq:N2}
	\begin{aligned}
	\text{Im}\int_\Gamma  y_\mu^* \bar{V}_\mu^*d\bar{V}_\mu&=\text{Im}\int_\Gamma  (g_\mu-jb_\mu)(x-jy)d(x+jy)\\
	&=\int_\Gamma  (-b_\mu x-g_\mu y)dx+(g_\mu x-b_\mu y)dy
	\end{aligned}
	\end{equation}
By Green's theorem, the integral above is path-independent if and only if
\begin{equation*}
	\frac{\partial(g_\mu x-b_\mu y)}{\partial x}=\frac{\partial(-b_\mu x-g_\mu y)}{\partial y}
	\Leftrightarrow g_\mu=-g_\mu\Leftrightarrow g_\mu=0
\end{equation*}
Similarly, we can prove that the real part of \eqref{eq:SLI} is path-independent if and only if all $b_\mu=0$.
\end{proof}
\begin{remark}
	In the context of the power system, the preceding proposition states that the integral \eqref{eq:N2} is path-independent if and only if we neglect all the transfer conductances or the susceptances in the network. Since the integral \eqref{eq:N2} is related to the energy function in the transient stability analysis of power systems \cite{Moon1}, the claim here is consistent with the common belief that the existence of non-zero transfer conductances leads to path-dependent term in energy functions.
\end{remark}
\begin{assumption}\label{as:3}
	We assume the transmission lines in the power system are lossless, i.e. $g_\mu=0$, $\forall\mu\in\mathcal{S}_1$.
\end{assumption}
Under this assumption, we have
\begin{equation}\label{eq:N3}
\begin{aligned}
\sum_{\mu\in\mathcal{S}_1}\text{Im}\int_\Gamma  \bar{I}_\mu^*d\bar{V}_\mu=\sum_{\mu\in\mathcal{S}_1}\int_\Gamma  (-b_\mu x)dx+(-b_\mu y)dy\\=\sum_{\mu\in\mathcal{S}_1}-\frac{1}{2}b_\mu(x^2+y^2)\Big|_\Gamma=\sum_{\mu\in\mathcal{S}_1}-\frac{1}{2}b_\mu|\bar{V}_\mu|^2\Big|_\Gamma
\end{aligned}
\end{equation}
which equals half of reactive power in the transmission lines.

For constant power component, by \eqref{eq:S} we have
\begin{equation}\label{eq:S2LI}
\im\int_\Gamma  \bar{I}_\mu^*d\bar{V}_\mu=-\im\int_\Gamma  \frac{P^0_\mu+jQ^0_\mu}{\bar{V}_\mu}d\bar{V}_\mu,\;\forall\mu\in\mathcal{S}_2
\end{equation}
The line integral is path-independent and if we denote $\bar{V}_\mu=V_\mu\angle\theta_\mu$, the integral can be express as
\begin{equation}\label{eq:S2LI1}
\im\int_\Gamma  \bar{I}_\mu^*d\bar{V}_\mu
=-P_\mu^0\theta_\mu-Q_\mu^0\ln V_\mu\Big|_\Gamma
\end{equation}

Now we are ready to define the voltage potential of a phasor circuit.
\begin{definition}
	Consider a phasor circuit satisfying Assumption \ref{as:3}. The function $V_p:\mathbb{C}^b\to\rr$ is called the voltage potential of the phasor circuit.
	\begin{equation}\label{eq:vp}
	\begin{aligned}
	V_p(\bar{V})&=\sum_{\mu\in\mathcal{S}}\text{Im}\int_\Gamma  \bar{I}_\mu^*d\bar{V}_\mu=\sum_{\mu\in\mathcal{S}_1}-\frac{1}{2}b_\mu|\bar{V}_\mu|^2\\&-\sum_{\mu\in\mathcal{S}_2}P_\mu^0\theta_\mu-Q_\mu^0\ln V_\mu
	\end{aligned}
	\end{equation}
\end{definition}
\subsection{Line Integral of Dynamic Components}
By Assumption \ref{as:2}, the end node of dynamic component $\mu$ is the ground. Suppose its non-ground terminal is indexed by $i$, then we have $\bar{V}_\mu=V_i\angle\theta_i$, and $P_\mu$ and $Q_\mu$ are identical to the power generated from the node i.e. $P_\mu+jQ_\mu=P_i+jQ_i$.

So for dynamic components, we have 
\begin{equation}\label{eq:DLI}
\begin{aligned}
\im\int_\Gamma\bar{I}_\mu^*d\bar{V}_\mu&=-\im\int_\Gamma\frac{P_i+jQ_i}{V_ie^{j\theta_i}}d(V_ie^{j\theta_i})\\&=-\int_\Gamma P_id\theta_i+Q_id\ln V_i
\end{aligned}
\end{equation}
This integral is generally path-dependent and is dictated by differential equations \eqref{eq:dynamic}.

Now we are ready to present the key observation in this paper. The proof is straightforward by invoking Lemma \ref{lem:line} and is omitted due to space limit.
\begin{theorem}\label{th:4}
	Consider a phasor circuit \eqref{eq:entire} satisfying Assumption \ref{as:1}-\ref{as:3}. For any solution trajectory  $\gamma:\rr_{\geq0}\to\mathbb{C}^b$, we have
	\begin{equation}\label{eq:thm4}
	\sum_{i\in\cal D}\int_{\gamma(0)}^{\gamma(t)} P_id\theta_i+Q_id\ln V_i=V_p(\gamma(t))-V_p(\gamma(0))
	\end{equation}
	where $V_p$ is the voltage potential of this phasor circuit.
\end{theorem}
\section{Distributed Stability Analytics}\label{sec:4}
\subsection{Distributed Criteria for System-Wide Stability}
In this section, we propose distributed criteria for system-wide stability  based on the voltage potential of phasor circuits.

By law of cosines, the modulus of branches voltage can be rewritten in terms of node voltages. Assume $\bar{V}_\mu=V_i\angle\theta_i-V_k\angle\theta_k$, we have
\begin{equation}\label{eq:bus}
|\bar{V}_\mu|^2=V_i^2+V_k^2-2V_iV_k\cos\theta_{ik}
\end{equation}
Substituting into \eqref{eq:N3} yields
\begin{equation}\label{eq:N4}
\text{Im}\int_\Gamma  \bar{I}_\mu^*d\bar{V}_\mu=\frac{1}{2}B_{ik}(V_i^2+V_k^2-2V_iV_k\cos\theta_{ik})\Big|_\Gamma
\end{equation}
where $B_{ik}$ is element in the network admittance matrix and we have $B_{ik}=-b_{ik}$.

Let $x^e:=\text{col}(\xi^e,V^e,\theta^e)$ be an equilibrium of the phasor circuit system \eqref{eq:entire}. For the voltage potential \eqref{eq:vp}, define the initial point $\bar{V}_0=V_0\angle\theta_0\in\mathbb{C}^b$ corresponding to $x^e$,\footnote{By Assumption \ref{as:2} all algebraic variables is dictated by state variables.} New results: it should be the initial point of the line integral, rather than the equilibrium. and formulate in the bus voltage coordinate \eqref{eq:N4}
\begin{equation}
\begin{aligned}
V_p(V,\theta)&=\sum_{(i,k)\in\mathcal{S}_1 }\frac{1}{2}B_{ik}(V_i^2+V_k^2-2V_iV_k\cos\theta_{ik})\\
&-\sum_{(i,k)\in\mathcal{S}_2}P_{ik}^0\theta_{ik}-Q_{ik}^0\ln V_{ik}
\end{aligned}
\end{equation}
Now consider the Bregman divergence \cite{bd} between $z:=\text{col}(V,\theta)$ and $z_0:=\text{col}(V_0,\theta_0)$ w.r.t. $V_p$ as follows.
\begin{equation}\label{eq:bd}
W(z)=V_p(z)-(z-z_0)^T\nabla V_p(z_0)-V_p(z_0)
\end{equation}
We have $W(z_0)=\nabla W(z_0)=0$ and $\nabla^2 W(z)=\nabla^2 V_p(z)$. Further, the Bregman divergence induced by $V_p$ has the following property.
\begin{lemma}\label{th:bd}
	Consider a phasor circuit \eqref{eq:entire} satisfying Assumption \ref{as:1}-\ref{as:3}. For any solution trajectory  $\gamma:\rr_{\geq0}\to\mathbb{C}^b$, we have
	\begin{equation}\label{eq:thmbd}
	\sum_{i\in\cal D}\int_{\gamma(0)}^{\gamma(t)} \Delta P_id\theta_i+\Delta Q_id\ln V_i=W(\gamma(t))-W(\gamma(0))
	\end{equation}
	where $\Delta P_i=P_i-P_i^e$, $\Delta Q_i=Q_i-Q_i^e$, and $W$ is the Bregman divergence induced by the voltage potential $V_p$.
\end{lemma}
\begin{theorem}\label{th:osci}
	Consider a phasor circuit \eqref{eq:entire} satisfying Assumption \ref{as:1}-\ref{as:3}. For any bounded trajectory  $\gamma:\rr_{\geq0}\to\mathbb{C}^b$, if 
	\begin{equation}\label{eq:thmoci}
	\int_{\gamma(0)}^{\gamma(t)} \Delta P_id\theta_i+\Delta Q_id\ln V_i\leq0,\;\forall i\in\mathcal{D},\forall t\in\rr_{\geq0}.
	\end{equation}
	and the largest invariant set of $\{\Delta P_id\theta_i+\Delta Q_id\ln V_i=0\}$ only contains equilibrium, then the trajectory $\gamma(t)$ will converge to the set of equilibrium.
\end{theorem}

Since the largest invariant set condition in Theorem \ref{th:osci} is usually satisfied for power system models \cite{ds}, Theorem \ref{th:osci} indicates that the unstable patterns of bounded trajectory, such as oscillation, can be precluded if every dynamic component obeys the integral inequality \eqref{eq:thmoci}. 

Consider again the Bregman divergence $W(z)$ \eqref{eq:bd}. We denote the set of equilibrium points which satisfy the convexity condition 
\begin{equation}\label{eq:convex}
E=\left\lbrace x^e:\nabla^2W(z_0)\geq0 \right\rbrace
\end{equation}
where the only zero eigenvalue of $\nabla^2W(z_0)$ comes from the rotational symmetry of phase angles \cite{dorfler}.

The following theorem offers a distributed criterion for stability w.r.t. any equilibrium in $E$.
\begin{theorem}\label{th:cs}
	Consider a phasor circuit \eqref{eq:entire} satisfying Assumption \ref{as:1}-\ref{as:3}. For any $x^e\in E$, if for all $i\in\mathcal{D}$, there exists a continues differentiable scalar function $W_i(x_i)$ such that $W_i$ is locally positive definite at local equilibrium $x_{i}^e$ and satisfies
\begin{equation}\label{eq:conx}
	\dot{W}_i\leq\Delta P_i\dot{\theta}_i+\Delta Q_i \dot{\ln V}_i
\end{equation}
	then the system-wide equilibrium $x^e$ is stable.
\end{theorem}

Note that in both Theorem \ref{th:osci} and Theorem \ref{th:cs}, we do not specify the dynamic models but propose generic criteria which can accommodate to the heterogeneity. Moreover, our criteria involve only local information as shown in \eqref{eq:thmoci} and \eqref{eq:conx}. Thus, it can be employed and assessed individually which fulfills the scalability requirement.
\begin{remark}
	Compared to the classical passivity condition \cite{passivity} , the left-hand side of \eqref{eq:conx} can be regarded as a supply rate, in which the input is $(\Delta P_i, \Delta Q_i)$, however, the output is time derivatives $(\dot{\theta},\dot{\ln V})$. Thus, the condition \eqref{eq:conx} is called a passivity-like condition in this paper.
\end{remark}
\begin{remark}
	The convexity condition \eqref{eq:convex} of equilibrium plays an important role in Theorem \ref{th:cs}. It guarantees a well-defined distance such that the Lyapunov argument can be employed. Roughly speaking, it is satisfied when the load is light. See \cite{Dj} for more information about under what condition the power system satisfies \eqref{eq:convex}.
\end{remark}
\subsection{Examples of Dynamic Components}
In this section, we give two specific examples of dynamic components and demonstrate how they meet the criterion in Theorem \ref{th:cs}.

We first consider the inverter-interfaced renewable energy sources which are controlled by the virtual synchronous generator (VSG) technique \cite{vsg} as follows.
\begin{equation}\label{eq:vsg}
\left\lbrace 
\begin{aligned}
\dot{\theta}_i&=\omega_i\\
M_i\dot{\omega}_i&=-D_i^p\omega_i+P^e_i-P_i\\
\tau_i^q\dot{V}_i&=-(V_i-V_i^e)-D_i^q(Q_i-Q_i^e)
\end{aligned}\right. 
\end{equation}
where $D_i^p$ is the droop coefficient, $\tau_i^q$ is the time constant, $M_i$ is the virtual inertial, and the superscript $e$ stands for the equilibrium value.

To meet the criterion in Theorem \ref{th:cs}, one can choose
\begin{equation}\label{eq:WVSG}
W_i(x_i)=\frac{1}{2}M_i\omega_i^2+\frac{k_i}{D_{i}^q}\left( \frac{V_i}{V^e_i}-\ln V_i\right)
\end{equation}
where $k_i=V_i^e+D_i^qQ_i^e$ is a constant. One can verify that condition \eqref{eq:conx} holds and $x_i^e=\text{col}(0,V_i^e,\theta_i^e)$ is a local minimum of \eqref{eq:WVSG} when $k_i>0$.

Another example is the inverter-interfaced component with the droop controller \cite{xie} as follows.
\begin{equation}\label{eq:droop}
\left\lbrace 
\begin{aligned}
\tau_{i}^p\dot{\theta}_i&=-(\theta_i-\theta_i^e)-D_{i}^p(P_i-P_i^e)\\
\tau_{i}^q\dot{V}_i&=-(V_i-V_i^e)-D_{i}^q(Q_i-Q_i^e)
\end{aligned}\right. 
\end{equation}
where $D_i^p,D_i^q$ are droop coefficients, $\tau_i^p,\tau_i^q$ are time constants, and the superscript $e$ stands for the equilibrium value.
To meet the criterion in Theorem \ref{th:cs}, similarly one can choose
\begin{equation}\label{eq:Wdroop}
W_i(x_i)=\frac{(\theta_i-\theta_i^e)^2}{2D_{i}^p}+\frac{k_i}{D_{i}^q}\left( \frac{V_i}{V^e_i}-\ln V_i\right)
\end{equation}
where $k_i=V_i^e+D_i^qQ_i^e$ and one can verify that condition \eqref{eq:conx} holds and $x_i^e=\text{col}(V_i^e,\theta_i^e)$ is a local minimum of \eqref{eq:Wdroop} when $k_i>0$.
\section{Case Study}
Consider a 3-bus power system as showed in Figure~\ref{fig1}. Bus 1 and 2 are attached to a VSG \eqref{eq:vsg} and a droop controlled \eqref{eq:droop} inverter source, respectively. Bus 3 is connected to a constant power load $P_3^0+jQ_3^0$. The parameters and equilibrium of the system are listed in Table~\ref{tab1} and Table~\ref{tab2}.
\begin{figure}[h]
	\centering
	\includegraphics[width=0.67\hsize]{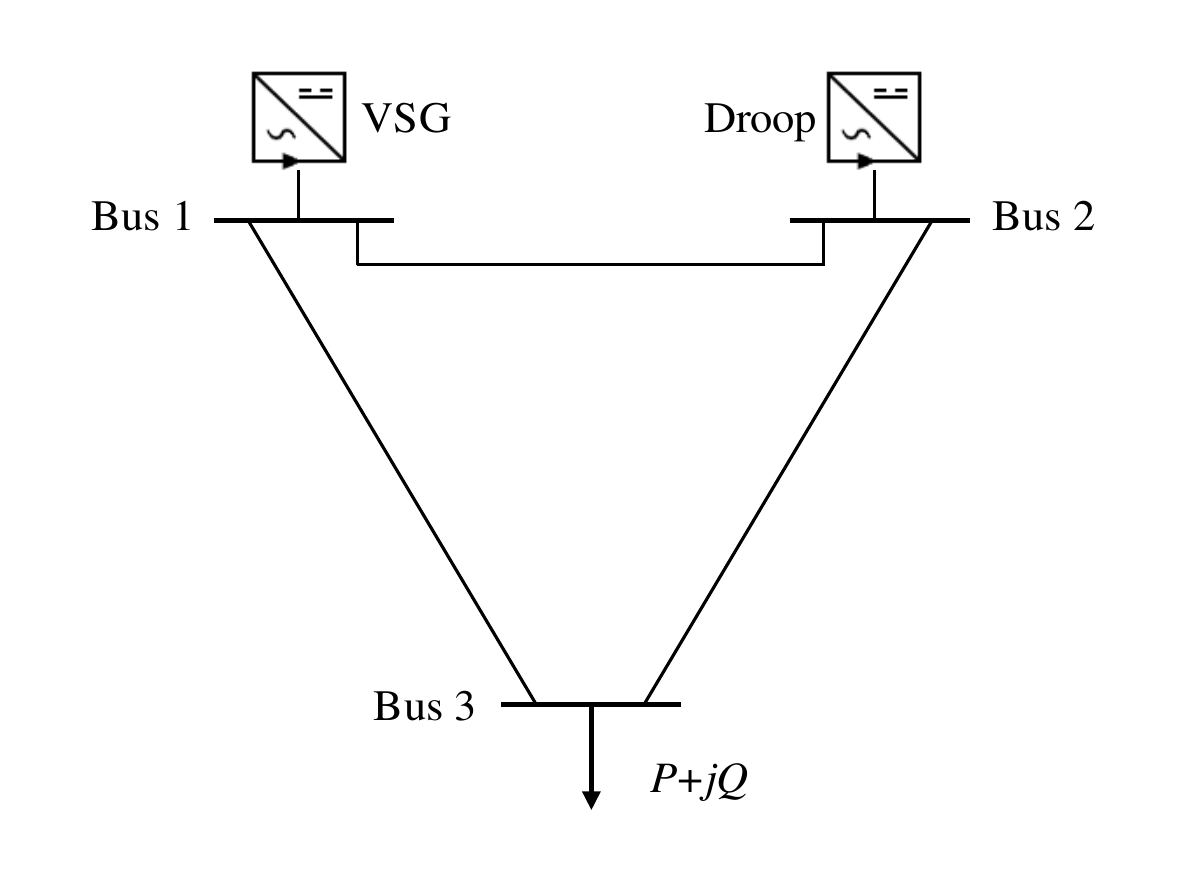}
	\caption{The schematic of the 3-bus system.}
	\label{fig1}
\end{figure}

One can verify that the equilibrium satisfies the convexity condition \eqref{eq:convex}. And both dynamic components meet the criterion in Theorem \ref{th:cs} with the given parameters. Thus, by Theorem \ref{th:cs}, it follows that the interconnected power system should be stable.
\begin{table}[!h]
	\centering
	\caption{System Parameters}
	\label{tab1}
	\begin{tabular}{l|l}
		\hhline
		Line Reactance         & 0.12  \\ \hline
		Virtual Inertial $M_1$   & 0.16 \\ \hline
		Droop Coefficients $D_1^p,D_1^q,D_2^p,D_2^q$   & 0.076, 0.03, 0.02, 0.02\\ \hline
		Time Constants $\tau_1^q,\tau_2^p,\tau_2^q$         & 0.3, 6.56, 8 \\ \hline
		Load Profile $P_3^0, Q_3^0$        & 0.03, 0.55 \\
		\hhline
	\end{tabular}
\end{table}
\begin{table}[!h]
	\centering
	\caption{System Equilibrium}
	\label{tab2}
	\begin{tabular}{c|c|c}
		\hhline
 $V_1^e\angle\theta_1^e$ & $V_2^e\angle\theta_2^e$ & $V_3^e\angle\theta_3^e$\\ \hline $1\angle0$ & $0.97\angle0.001$& $0.95\angle-0.0015$ \\
		\hhline
	\end{tabular}
\end{table}

To verify the theoretical result, suppose the system encounters a fault and undergoes a transient process. The dynamic response is showed in Figure~\ref{fig2} (a) and (b). The system is stable which is consistent with our claim. The voltage potential $V_p$ \eqref{eq:vp} is also depicted in Figure~\ref{fig2} (c). It is clear that $V_p$ tends to zero as the system converges to the stable equilibrium.

 \begin{figure}[!h]
 	\centering
 	\includegraphics[width=\hsize]{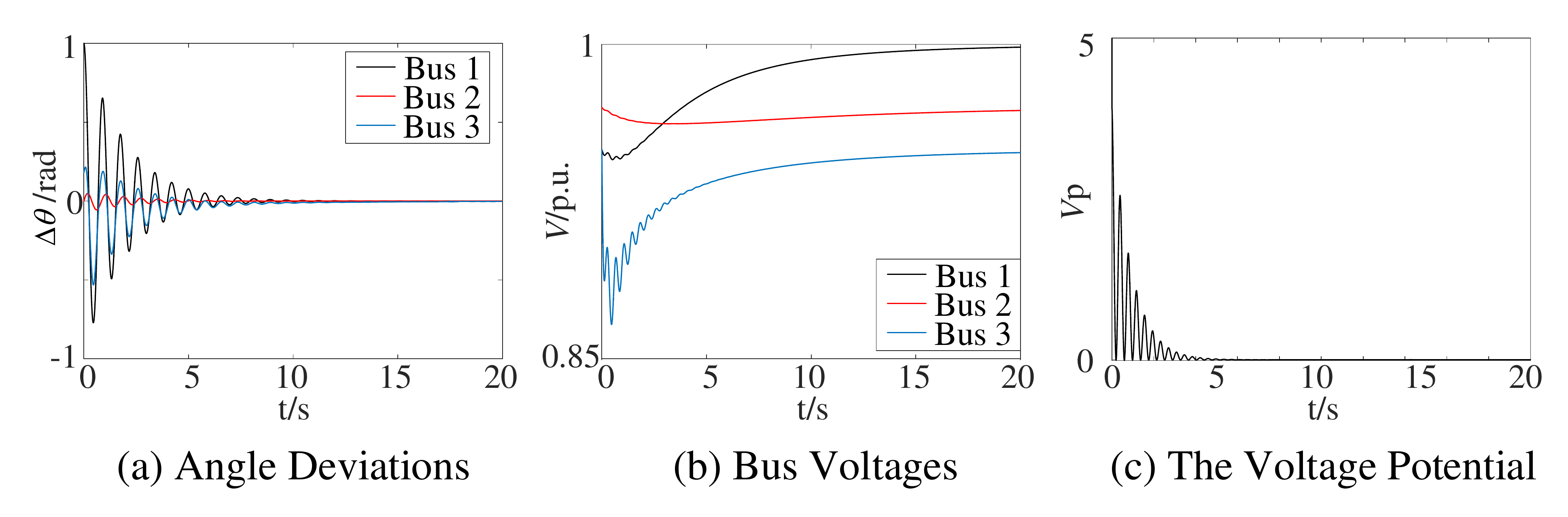}
 	\caption{(a) The angle deviations, (b) the voltage deviations, and (c) the voltage potential during the transient process.}
 	\label{fig2}
 \end{figure}

\section{Concluding Remarks}
We have presented a phasor-circuit theory perspective to handle the stability issue of power systems. Based on the observation that the symmetric AC three-phase power system can be regarded as a phasor circuit, we have extended and studied the concept of voltage potential with mathematical tools from complex analysis. Our results show that under the convexity condition, the system-wide stability can be ensured if each dynamic component meets a passivity-like condition, which can fit heterogeneous models and is scalable.

In future works, we will relax the convexity condition and enlarge the valid scope of our criteria. We believe that when the power system becomes more complex, in order to handle its stability issues, it is helpful or even necessary to review some basic circuit theories.

\end{document}